\documentclass[11pt]{article}

\usepackage[utf8]{inputenc}
\usepackage[T1]{fontenc}
\usepackage{lmodern}
\usepackage{geometry}
\usepackage{amsmath,amssymb,amsthm,mathtools}
\usepackage{microtype}
\usepackage{hyperref}
\hypersetup{colorlinks=true, linkcolor=blue, citecolor=blue, urlcolor=blue}
\usepackage{algorithm}
\usepackage{algpseudocode}
\usepackage{enumitem}
\newtheorem{theorem}{Theorem}

\newtheorem{lemma}[theorem]{Lemma}

\theoremstyle{remark}
\newtheorem{remark}[theorem]{Remark}
\newtheorem{example}[theorem]{Example}

\newcommand{\F}{\mathbb{F}}

\newcommand{\wtsr}{\mathrm{wt}_{\mathrm{sr}}}
\newcommand{\wtH}{\mathrm{wt}_{\mathrm{H}}}
\newcommand{\supp}{\mathrm{supp}}
\newcommand{\SR}{\mathrm{SR}}

\title{\vspace{-1em}Two-Step Decoding of Binary $2\times2$ Sum-Rank-Metric Codes}
\author{Hao Wu$^1$, Bocong Chen$^1$, Guanghui Zhang$^2$ and Hongwei Liu$^3$\footnote{E-mail addresses:
mawuhao\_math@mail.scut.edu.cn (H. Wu);
bocongchen@foxmail.com (B. Chen);
zghui@squ.edu.cn (G. Zhang); hwliu@ccnu.edu.cn (H. Liu).}}
\date{\small
$1.$ School of Mathematics, South China University of Technology, Guangzhou, 510641, China\\
$2.$ School of Mathematics and Physics, Suqian University, Suqian, Jiangsu 223800, China\\
$3.$ School of Mathematics and Statistics, Central China Normal University, Wuhan 430079, China
}

\begin{document}
\maketitle
\begin{abstract}
We address an open problem posed by Chen-Cheng-Qi (IEEE Trans.\ Inf.\ Theory, 2025):
can the decoding of binary sum-rank-metric codes $\SR(C_1,C_2)$ with $2\times2$ matrix blocks be reduced entirely to decoding the constituent Hamming-metric codes $C_1$ and $C_2$ without the additional requirement $d_1\ge\tfrac{2}{3}d_{\mathrm{sr}}$ used in their fast decoder?
We answer this in the affirmative by exhibiting a simple two-step procedure: first uniquely decode $C_2$, then apply a single error-erasure decoding for $C_1$.
This shows that the restrictive hypothesis
$d_1\ge\tfrac{2}{3}d_{\mathrm{sr}}$ is theoretically unnecessary.
The resulting decoder achieves unique decoding up to
$\lfloor (d_{\mathrm{sr}}-1)/2\rfloor$ with overall cost $T_2+T_1$, where $T_2$ and $T_1$ are the complexities of the Hamming decoders for $C_2$ and $C_1$, respectively.
We further show that this reduction is asymptotically optimal in a black-box model, as any sum-rank decoder must inherently decode the constituent Hamming codes.
For BCH or Goppa instantiations over $\F_4$, the decoder runs in
$O(\ell^2)$ time.
\end{abstract}


\section{Introduction}
Let $q$ be a prime power and let $\ell\ge1$ be an integer.
Let $\F_q$ be the finite field with $q$ elements.
For integers $m,n\ge1$, denote by $\F_q^{m\times n}$ the
set of all $m\times n$ matrices with entries in $\F_q$.
For positive integers $m_1,\dots,m_\ell$ and $n_1,\dots,n_\ell$, consider the ambient space
\[
\mathcal{V}
=\F_q^{m_1\times n_1}\times\cdots\times \F_q^{m_\ell\times n_\ell}.
\]
For a vector $X=(X_1,\dots,X_\ell)\in\mathcal{V}$ with $X_i\in\F_q^{m_i\times n_i}$, its \emph{sum-rank weight} is defined as
\[
\wtsr(X)\;=\;\sum_{i=1}^\ell \mathrm{rk}_q(X_i),
\]
where $\mathrm{rk}_q(\cdot)$ denotes the matrix rank over $\F_q$. The associated \emph{sum-rank distance} between $X,Y\in\mathcal{V}$ is
\[
d_{\mathrm{sr}}(X,Y)\;=\;\wtsr(X-Y).
\]
This defines a metric on $\mathcal{V}$.

A $q$-ary \emph{sum-rank-metric code} $C$ with block length $\ell$ and matrix sizes $m_1\times n_1, m_2\times n_2, \cdots, m_\ell\times n_\ell$ is a subset of the
space $\mathcal{V}$. For a $q$-ary sum-rank-metric code (not necessarily linear) $C\subseteq\mathcal{V}$, one defines its minimum sum-rank distance by
\[
d_{\mathrm{sr}}(C)
\;=\;
\min_{\substack{X,Y\in C\\X\neq Y}} d_{\mathrm{sr}}(X,Y).
\]
When $C$ is $\F_q$-linear, this simplifies to
\[
d_{\mathrm{sr}}(C)
=\min_{0\neq Z\in C} \wtsr(Z).
\]

In this general framework,
Hamming metric arises when all blocks are $1\times1$, while rank metric arises when $\ell=1$.
More generally, sum-rank-metric codes allow several matrix blocks and simultaneously capture Hamming-type and rank-type behavior.

Sum-rank-metric codes form a natural bridge between classical Hamming-metric codes and rank-metric codes, and have found applications in multishot network coding \cite{MK, NPS, NU},
distributed storage \cite{CMST, MK2018, MP2018}, and space-time coding \cite{SK2022}.
From a theoretical perspective, sum-rank-metric codes inherit structural
features from both Hamming-metric and rank-metric codes and admit analogues
of classical notions such as minimum distance, Singleton bounds, and
maximum-distance-separable (MDS) constructions.
In recent years, many papers have studied the structure, constructions, and applications of sum-rank-metric codes;
see \cite{BGR, BGR2022, BZ, MK, MP2018, MP2020, MP2021, MP2022, MP2023, NERI2022, NERI2023} and the references therein.
In addition, we refer the reader to \cite{AGKP, AKR, BGR}  for fundamental properties and some bounds on sizes of sum-rank-metric codes
and to \cite{MSK2022} for a nice survey of sum-rank-metric codes and their applications.

A central problem in the theory of sum-rank codes is the design of efficient
decoding algorithms.
Several constructions and decoding strategies have been developed in recent
years; see, for example,
\cite{BJPR2021, HBP2022, PRR2022}.
However, in many settings the decoding procedures still rely on nontrivial
algebraic operations over extension fields or specialized code structures.
This raises a natural question: to what extent can decoding in the
sum-rank metric be reduced to decoding simpler constituent codes
in the Hamming metric?

Recently, Chen, Cheng, and Qi \cite{CCQ2025} studied binary
sum-rank-metric codes of the form $SR(C_1,C_2)$ with $2\times2$
matrix blocks. They gave an explicit and flexible construction based on two quaternary codes and
showed that this family already exhibits many of the features expected of general sum-rank-metric codes, such as good parameters,
algebraic structure, and efficient decoding.
In particular, they proposed a fast decoding algorithm whose correctness relies on the
additional assumption $d_1 \ge \tfrac{2}{3}d_{\mathrm{sr}}$,
where $d_1$ denotes the minimum Hamming distance of $C_1$
and $d_{\mathrm{sr}}$ is the minimum sum-rank distance of the code.
They posed the question of whether this extra constraint can be eliminated.
Specifically, they asked whether the decoding problem for the sum-rank-metric code can be reduced solely to
the decoders of $C_1$ and $C_2$, thereby removing the assumption $d_1 \ge \tfrac{2}{3}d_{\mathrm{sr}}$.

In this paper, we answer this question affirmatively.
We show that decoding of binary $2\times2$ sum-rank-metric codes
$SR(C_1,C_2)$ can be completely reduced to decoding the constituent
Hamming-metric codes $C_1$ and $C_2$, without requiring the condition
$d_1 \ge \tfrac{2}{3}d_{\mathrm{sr}}$.

\vspace{0.3cm}
{\bf A}.  \emph{The Chen-Cheng-Qi construction $\SR(C_1,C_2)$}
\vspace{0.3cm}

Fix the quadratic extension $\F_4/\F_2$ and an $\F_2$-basis $\{1,\omega\}$ of $\F_4$ with $\omega^2+\omega+1=0$.
As in~\cite{CCQ2025}, the starting point is a concrete
identification between pairs $(x_1,x_2)\in\F_4^2$
and $2\times2$ binary matrices.
For $x_1,x_2\in\F_4$, consider the \emph{$2$-polynomial}, i.e., linearized polynomial
\[
L_{x_1,x_2}(x)\;=\;x_1 x + x_2 x^2, x\in\F_4.
\]
Since $x\mapsto x^2$ is the Frobenius automorphism of $\F_4$ over $\F_2$, the map $L_{x_1,x_2}:\F_4\to\F_4$ is $\F_2$-linear. With respect to the fixed basis $\{1,\omega\}$, this $\F_2$-linear map is represented by a unique $2\times2$ binary matrix
\[
M(x_1,x_2)\;\in\;\F_2^{2\times2}.
\]
In this way we obtain an explicit $\F_2$-linear bijection
\[
\Phi:\F_4^2\longrightarrow \F_2^{2\times2},
\quad (x_1,x_2)\longmapsto M(x_1,x_2).
\]
A length-$\ell$ sequence of $2\times2$ matrices is then written as
\[
\mathbf{a}_1x+\mathbf{a}_2 x^2,\quad
\mathbf{a}_1=(a_{11},\cdots,a_{1\ell}),~
\mathbf{a}_2=(a_{21},\cdots,a_{2\ell})\in\F_4^\ell,
\]
where the $i$-th coordinate pair $(a_{1,i},a_{2,i})$ is sent to the matrix $M(a_{1,i},a_{2,i})$ via $\Phi$. Given  linear  quaternary codes $C_1=[\ell,k_1,d_1]_4$ and $C_2=[\ell,k_2,d_2]_4$, Chen-Cheng-Qi define the associated binary sum-rank-metric code with matrix size $2\times2$ by
\[
\SR(C_1,C_2)
\;=\;
\bigl\{\,\mathbf{a}_1x+\mathbf{a}_2x^2 : \mathbf{a}_1\in C_1,\ \mathbf{a}_2\in C_2\,\bigr\},
\]
that is, each codeword of $\SR(C_1,C_2)$ is an $\ell$-tuple of $2\times2$ binary matrices obtained by applying $\Phi$ coordinatewise to a pair
$(\mathbf{a}_1,\mathbf{a}_2)\in C_1\times C_2$.
It was shown that
$\SR(C_1,C_2)$ is a binary linear sum-rank-metric code of block length $\ell$, matrix size $2\times2$, and   dimension
\[
\dim_{\F_2}\SR(C_1,C_2)=2(k_1+k_2),
\]
see~\cite[Theorem 2.1]{CCQ2025}.
A key ingredient in~\cite{CCQ2025} is an  exact decomposition of the sum-rank weight of a codeword $\mathbf{a}_1x+\mathbf{a}_2x^2\in\SR(C_1,C_2)$ in terms of the Hamming weights of $\mathbf{a}_1$ and $\mathbf{a}_2$. Writing
\[
I=\supp(\mathbf{a}_1)\cap\supp(\mathbf{a}_2)\subseteq\{1,\dots,\ell\},
\]
Chen-Cheng-Qi prove the identity
\begin{equation}\label{eq:CCQ-weight}
\mathrm{wt}_{\mathrm{sr}}(\mathbf{a}_1x+\mathbf{a}_2x^2)
=2\,\mathrm{wt}_H(\mathbf{a}_1)+2\,\mathrm{wt}_H(\mathbf{a}_2)-3|I|,
\end{equation}
see~\cite[Theorem 2.2]{CCQ2025}; in particular,
\[
d_{\mathrm{sr}}(\SR(C_1,C_2))
\ \ge\ \min\{d_1,2d_2\}.
\]
Thus the pair $(C_1,C_2)$ controls both the dimension and the distance of $\SR(C_1,C_2)$ via classical Hamming-metric parameters.

\vspace{0.3cm}
{\bf B}. \emph{The fast decoder and its complexity}
\vspace{0.3cm}

In the second half of~\cite{CCQ2025}, the authors use the weight formula~\eqref{eq:CCQ-weight} to derive a fast decoding algorithm for $\SR(C_1,C_2)$, reducing sum-rank decoding to Hamming-metric decoding of $C_1$ and $C_2$. Let $d_{\mathrm{sr}}$ denote the minimum sum-rank distance of $\SR(C_1,C_2)$, and write the transmission as
\[
\mathbf{y}=\mathbf{c}+\mathbf{e},\qquad \mathbf{c}=\mathbf{a}_1x+\mathbf{a}_2x^2\in\SR(C_1,C_2),\quad \mathbf{e}=\mathbf{e}_1x+\mathbf{e}_2x^2.
\]
Write $\mathbf{e}_1=(e_{1,1},\dots,e_{1,\ell})$ and $\mathbf{e}_2=(e_{2,1},\dots,e_{2,\ell})$. Chen-Cheng-Qi partition the coordinates into three types according to the pattern of $(e_{1,i},e_{2,i})$:
\[
I_1=\{\,i : e_{1,i}\neq0,\ e_{2,i}=0\,\},\quad
I_2=\{\,i : e_{1,i}=0,\ e_{2,i}\neq0\,\},\quad
I_3=\{\,i : e_{1,i}\neq0,\ e_{2,i}\neq0\,\},
\]
and denote $i_j=|I_j|$ for $j=1,2,3$. A direct inspection of the $2\times2$ rank in each case gives
\[
\wtsr(\mathbf{e})=2i_1+2i_2+i_3,
\]
so in particular
\[
\mathrm{wt}_H(\mathbf{e}_2)=i_2+i_3\ \le\ \wtsr(\mathbf{e}).
\]
Thus, under
\[
\wtsr(\mathbf{e})\le\Big\lfloor\frac{d_{\mathrm{sr}}-1}{2}\Big\rfloor
\quad\text{and}\quad
d_2\ge d_{\mathrm{sr}},
\]
a bounded-minimum-distance decoder for $C_2$ can uniquely recover $\mathbf{a}_2$ and $\mathbf{e}_2$ from the second component $\mathbf{y}_2=\mathbf{a}_2+\mathbf{e}_2$.

The more delicate part is recovering $\mathbf{a}_1$ from $\mathbf{y}_1=\mathbf{a}_1+\mathbf{e}_1$ when some coordinates have both $e_{1,i}$ and $e_{2,i}$ nonzero. After subtracting $\mathbf{a}_2 x^2$ and rewriting the received word as
\[
\mathbf{y}'=\mathbf{y}-\mathbf{a}_2 x^2=(\mathbf{a}_1+\mathbf{e}_1)x+\mathbf{e}_2 x^2,
\]
one can view, at each $i\in I_3$, the local error as the linearized polynomial
\[
x\longmapsto e_{1,i}x+e_{2,i}x^2.
\]
Since this polynomial has rank~1, it has a nonzero root in $\F_4^\times=\{1,\omega,\omega^2\}$. This motivates evaluating $\mathbf{y}'$ at $x=1,\omega,\omega^2$ to obtain three quaternary words that, after a suitable coordinatewise scaling by
\[
\mathbf{1}=(1,\dots,1),\quad
\mathbf{u}=(\omega,\dots,\omega),\quad
\mathbf{u}^2=(\omega^2,\dots,\omega^2)\in\F_4^\ell,
\]
lie in $C_1$, $\mathbf{u}\cdot C_1$ and $\mathbf{u}^2\cdot C_1$, respectively. For each $i\in I_3$, at least one of these three evaluations cancels the local error, so the corresponding Hamming error pattern is strictly lighter. A simple averaging or counting argument then shows that, under the additional condition
\[
d_1\ \ge\ \tfrac{2}{3}d_{\mathrm{sr}},
\]
at least one of the three derived words lies within the Hamming decoding radius $\lfloor(d_1-1)/2\rfloor$ of $C_1$ (or of $\mathbf{u}\cdot C_1$ or $\mathbf{u}^2\cdot C_1$). The decoder therefore runs a bounded-minimum-distance decoder for $C_1$ on each of the three candidates and selects the uniquely consistent output, thereby recovering $\mathbf{a}_1$.

In summary, under
\[
d_2\ge d_{\mathrm{sr}}
\qquad\text{and}\qquad
d_1\ge\tfrac{2}{3}d_{\mathrm{sr}},
\]
their Theorem~4.1 shows that one can uniquely decode $\SR(C_1,C_2)$ up to $\lfloor(d_{\mathrm{sr}}-1)/2\rfloor$ sum-rank errors by one call to a decoder for $C_2$ followed by three calls to (equivalents of) a decoder for $C_1$. If $T_2$ and $T_1$ denote the complexities of the bounded-minimum-distance decoders for $C_2$ and $C_1$, respectively, the overall complexity is
\[
T_{\mathrm{CCQ}}=T_2+3T_1.
\]
When $C_1$ and $C_2$ are BCH or Goppa codes over $\F_4$, one has $T_1(\ell),T_2(\ell)=O(\ell^2)$, so the resulting sum-rank decoder also has complexity $O(\ell^2)$ over $\F_4$; see, e.g.,~\cite[Sec.~IV]{CCQ2025}. In particular, Chen-Cheng-Qi obtain explicit families of binary $2\times2$ sum-rank-metric codes that are both quadratically encodable and decodable, at the price of the additional constraint $d_1\ge\tfrac{2}{3}d_{\mathrm{sr}}$.

\newpage

\vspace{0.3cm}
{\bf C}.  \emph{ The open problem of Chen-Cheng-Qi}
\vspace{0.3cm}

Although the decoder of~\cite{CCQ2025} already achieves the optimal unique-decoding radius $\lfloor(d_{\mathrm{sr}}-1)/2\rfloor$ and quadratic complexity, it requires the extra inequality $d_1\ge\tfrac{2}{3}d_{\mathrm{sr}}$, which strongly constrains code design. In their conclusion, Chen-Cheng-Qi explicitly raised the following open question:

\vspace{0.3cm}
\emph{Can the decoding of the binary $2\times2$ sum-rank-metric codes $\SR(C_1,C_2)$ from Theorem~2.1 in~\cite{CCQ2025} be reduced to the decoders of $C_1$ and $C_2$ \emph{without} assuming $d_1\ge\tfrac{2}{3}d_{\mathrm{sr}}$}?

\vspace{0.3cm}
{\bf D}. \emph{ Our contributions}
\vspace{0.3cm}

In this paper we answer the above question in the affirmative. Our main result shows that the additional inequality $d_1\ge\tfrac{2}{3}d_{\mathrm{sr}}$ can be completely removed while preserving both the unique-decoding radius $\lfloor(d_{\mathrm{sr}}-1)/2\rfloor$ and the quadratic-time complexity. Under the single assumption $d_2\ge d_{\mathrm{sr}}$, we give a particularly simple two-step reduction from sum-rank decoding of $\SR(C_1,C_2)$ to Hamming-metric decoding of $C_1$ and $C_2$.
To describe the procedure, write the received word as $\mathbf{y}=\mathbf{c}+\mathbf{e}$, where $\mathbf{c}=\mathbf{a}_1 x+\mathbf{a}_2 x^2\in\SR(C_1,C_2)$ and $\mathbf{e}=\mathbf{e}_1 x+\mathbf{e}_2 x^2$.

(1) First uniquely decode $C_2$ from the second component
  $\mathbf{y}_2=\mathbf{a}_2+\mathbf{e}_2$ up to $\lfloor(d_2-1)/2\rfloor$ Hamming errors, recovering $\mathbf{a}_2$ and the error vector $\mathbf{e}_2$;

(2) Then use the support $J=\supp(\mathbf{e}_2)$ as an erasure pattern and perform a single error/erasure decoding of $C_1$.

A direct application of the weight identity~\eqref{eq:CCQ-weight} and the classical condition $2t+r<d_1$ for Hamming error-erasure decoding shows that this procedure succeeds whenever $\wtsr(\mathbf{e})\le\lfloor(d_{\mathrm{sr}}-1)/2\rfloor$. In other words, the condition $d_1\ge\tfrac{2}{3}d_{\mathrm{sr}}$ in Theorem~4.1 of~\cite{CCQ2025} is not needed for unique decoding up to half the minimum sum-rank distance.

The contributions of this work can be summarized as follows.
\begin{itemize}[leftmargin=1.5em]
  \item Answer to the open problem of Chen-Cheng-Qi.
  We prove that for every pair of  linear   quaternary codes $C_1,C_2\subseteq\F_4^\ell$ with $d_2\ge d_{\mathrm{sr}}(\SR(C_1,C_2))$, the code $\SR(C_1,C_2)$ is uniquely decodable up to $\lfloor(d_{\mathrm{sr}}-1)/2\rfloor$ sum-rank errors by a reduction to the Hamming decoders of $C_1$ and $C_2$ that does \emph{not} require $d_1\ge\tfrac{2}{3}d_{\mathrm{sr}}$.
  If $T_2$ and $T_1$ denote the complexities of the Hamming decoders for $C_2$ and $C_1$, respectively, the overall cost of our reduction is
  \[
  T_{\mathrm{SR}}=T_2+T_1+O(\ell).
  \]
  For BCH or Goppa instantiations over $\F_4$ one has $T_1(\ell),T_2(\ell)=O(\ell^2)$, so $\SR(C_1,C_2)$ admits $O(\ell^2)$ decoding with a smaller constant than in~\cite{CCQ2025}, and with a significantly enlarged design region for $(d_1,d_2)$.

  \item  Simplified design condition.
If $d_2\geq 2d_1$, then the decoder is guaranteed to succeed.
Note that since $d_{sr} \le 2d_1$,    the condition
$d_2 \ge 2d_1$   serves as a sufficient condition for Theorem \ref{thm:SR-decoding}.
This provides a simple design
 criterion: if the second component code   $C_2$  is chosen to have at least twice the
 Hamming distance of  $C_1$,
    the proposed two-step decoder is guaranteed to succeed
up to the full unique decoding radius, without needing to compute the exact sum-rank
distance $ d_{sr}$.

  \item Black-box optimality.
  We show that, in a natural black-box model where the cost of decoding $\SR(C_1,C_2)$ is measured only via the complexities $T_1(\ell)$ and $T_2(\ell)$ of the underlying Hamming decoders, our two-step reduction is optimal up to constant factors: any algorithm that uniquely decodes $\SR(C_1,C_2)$ up to half its minimum sum-rank distance must, in particular, be able to decode certain embedded Hamming subcodes of $C_1$ and $C_2$ to
  their Hamming unique-decoding radius $\lfloor(d_1-1)/2\rfloor$ and $\lfloor(d_2-1)/2\rfloor$, respectively.
\end{itemize}

\vspace{0.3cm}
{\bf E.}  \emph{Organization }
\vspace{0.3cm}

The rest of this paper is organized as follows. In Section 2, we review the necessary background on
Hamming-metric codes and error-erasure decoding, including the classical uniqueness condition $2t + r < d$.
Section 3 presents our main result: a two-step decoding procedure for $\mathrm{SR}(C_1, C_2)$ that
removes the restrictive assumption $d_1 \geq \frac{2}{3}d_{\mathrm{sr}}$.
We provide a detailed proof and illustrate the method with a concrete example.
Section 4 analyzes the complexity of the proposed algorithm and compares it with the decoder of
Chen-Cheng-Qi, highlighting the expanded design region for the constituent codes. In Section 5,
we discuss the black-box optimality of our reduction,
showing that any sum-rank decoder must inherently decode the underlying Hamming codes. Section 6 concludes the paper.

\section{Preliminaries}\label{sec:prelim-Hamming}

In this section we recall the standard Hamming error--erasure model and the
classical uniqueness condition $2t+r<d$ for linear codes.
This is one of the main
decoding principles we will need: the sum-rank part of the argument is handled
by the exact weight identity~\eqref{eq:CCQ-weight}, while all decoding
reductions are carried out purely in the Hamming metric.

\subsection{Linear codes and Hamming distance}

Let $\F_q$ be a finite field. A \emph{code} of length $n$ over $\F_q$ is any
subset $C\subseteq \F_q^n$. We say that $C$ is \emph{linear} if it is an
$\F_q$-linear subspace of $\F_q^n$. In this case, if $\dim_{\F_q}C=k$, we call
$C$ an $[n,k]_q$ linear code.
For a vector $\mathbf{x}=(x_1,\dots,x_n)\in \F_q^n$, its \emph{Hamming weight} is
$$
\mathrm{wt}_{H}(\mathbf{x})\;=\;|\{\,i : x_i\neq 0\,\}|,$$
and its support is
$$
\supp(\mathbf{x})\;=\;\{\,i : x_i\neq 0\,\}.
$$
The \emph{Hamming distance} between $x,y\in\F_q^n$ is
\[
d_{H}(\mathbf{x},\mathbf{y})\;=\;\mathrm{wt}_{H}(\mathbf{x}-\mathbf{y}).
\]
For a (not necessarily linear) code $C\subseteq\F_q^n$, its minimum Hamming
distance is
\[
d_{H}(C)
\;=\;
\min_{\substack{\mathbf{x},\mathbf{y}\in C\\\mathbf{x}\neq \mathbf{y}}} d_{H}(\mathbf{x},\mathbf{y}),
\]
and for a linear code $C$ this simplifies to
\[
d_{H}(C)
\;=\;
\min_{0\neq \mathbf{x}\in C} \mathrm{wt}_{H}(\mathbf{x}).
\]
We will write $d$ for $d_{H}(C)$ when the code $C$ is clear from the
context, and $[n,k,d]_q$ denotes a linear code over $\F_q$ of length $n$,
dimension $k$ and minimum distance $d$.

\subsection{Transmission with errors and erasures}

In the classical Hamming model, we assume that a codeword $\mathbf{c}\in C$ is
transmitted over a noisy channel and a vector $\mathbf{y}\in\F_q^n$ is received. The
\emph{error vector} is $\mathbf{e}=\mathbf{y}-\mathbf{c}$. Its Hamming weight $\mathrm{wt}_H(\mathbf{e})$ is the number of
positions where the received symbol differs from the transmitted symbol.

The error-erasure model refines this as follows. Besides corrupting symbols,
the channel (or a higher-level protocol) may mark some positions as unreliable,
indicating that the symbol there is unknown. Formally, we are given:
\begin{itemize}[leftmargin=1.5em]
  \item a codeword $\mathbf{c}\in C$;
  \item an error vector $\mathbf{e}\in\F_q^n$;
  \item an erasure set $J\subseteq\{1,\dots,n\}$.
\end{itemize}
The received vector $\mathbf{y}\in\F_q^n$ is defined by
\[
y_i=
\begin{cases}
c_i+e_i, & i\notin J,\\[0.3em]
\text{``erasure''}, & i\in J.
\end{cases}
\]
In positions $i\notin J$, we observe a concrete symbol $y_i$ but do not know
whether it was corrupted ($e_i\neq 0$). In positions $i\in J$, we only know
that the symbol is unreliable (we can think of $y_i$ as a distinguished
``$\mathsf{?}$'' symbol). The \emph{erasure number} is $r=|J|$.

The genuine unknown error positions in this model are those outside $J$ where
$e_i\neq 0$:
\[
E=\{\,i\notin J : e_i\neq 0\,\},\qquad t=|E|.
\]
We will refer to $t$ as the \emph{error number}. Note that errors inside $J$
are not counted in $t$; their effect has already been absorbed into the
erasure set.

\subsection{Uniqueness condition $2t+r<d$}

The following uniqueness criterion for Hamming error-and-erasure decoding is classical; see, e.g.,
\cite[Ch.~7]{HuffmanPless}, \cite[Ch.~1]{MacWilliamsSloane}, \cite[Ch.~2]{vanLint}. We include a proof for completeness.

\begin{lemma}\label{lem:erasure}
Let $C\subseteq\F_q^n$ be a linear code with minimum Hamming distance $d$.
Let $J\subseteq\{1,\dots,n\}$ be the erasure set with $r=|J|$, and suppose the number of errors
outside $J$ is $t$. If $2t+r<d$, then in the punctured code $C|_{\overline{J}}$ (whose minimum
distance is at least $d-r$), there is \emph{at most one} codeword within Hamming distance $t$
of the punctured received word $y|_{\overline{J}}$.
\end{lemma}
\begin{proof}
Puncture $C$ at $J$ to obtain $C'$ with $d(C')\ge d-r$. If two codewords $\mathbf{u}',\mathbf{v}'\in C'$
satisfy $d_H(\mathbf{u}',\mathbf{y}')\le t$ and $d_H(\mathbf{v}',\mathbf{y}')\le t$ for $\mathbf{y}'=\mathbf{y}|_{\overline{J}}$, then
$d_H(\mathbf{u}',\mathbf{v}')\le 2t<d(C')$, a contradiction.
\end{proof}

Lemma~\ref{lem:erasure} is purely a \emph{uniqueness} statement: it says that
if a decoder outputs any codeword consistent with the observed errors and
erasures, that codeword must be the transmitted one whenever $2t+r<d$. It does
not by itself specify how to construct such a decoder. In practice, many
algebraic decoders for BCH, Reed-Solomon, and Goppa codes can handle errors
and erasures by modifying the key-equation or syndrome equations appropriately;
see, for example,~\cite{HuffmanPless,MacWilliamsSloane,Sugiyama,vanLint}.

\section{Main Result and its Proof}\label{sec:main}

In this section we give a detailed  proof of the main decoding
result: under the single assumption $d_2\ge d_{\mathrm{sr}}$, the code
$\SR(C_1,C_2)$ can be uniquely decoded up to half its minimum sum-rank distance
by a two-step procedure that first uniquely decodes $C_2$ and then performs a
single error-erasure decoding of $C_1$ guided by the support of the recovered
error in the second component.


We adopt the notation and framework of~\cite{CCQ2025}. Let
$C_1=[\ell,k_1,d_1]_4$ and $C_2=[\ell,k_2,d_2]_4$ be quaternary codes and let
$\SR(C_1,C_2)$ denote the associated binary sum-rank-metric code of matrix size
$2\times 2$ as in the introduction. Recall that every codeword of $\SR(C_1,C_2)$
can be written as
\[
\mathbf{c}(x)=\mathbf{a}_1x+\mathbf{a}_2x^2,
\qquad \mathbf{a}_1,\mathbf{a}_2\in \F_4^\ell.
\]
We consider transmission of a codeword $\mathbf{c}(x)=\mathbf{a}_1x+\mathbf{a}_2x^2\in \SR(C_1,C_2)$ over a
binary sum-rank channel, with received word
\[
\mathbf{y}(x)=\mathbf{c}(x)+\mathbf{e}(x)=
(\mathbf{a}_1+\mathbf{e}_1)x+(\mathbf{a}_2+\mathbf{e}_2)x^2,
\]
where $\mathbf{e}(x)=\mathbf{e}_1x+\mathbf{e}_2x^2$ and $\mathbf{e}_1,\mathbf{e}_2\in \F_4^\ell$ are the error components.

Following Chen-Cheng-Qi~\cite{CCQ2025}, for each coordinate
$i\in\{1,\dots,\ell\}$ we inspect the error pair $(e_{1,i},e_{2,i})$ and classify
the index $i$ into one of three disjoint sets:
\begin{align*}
I_1 &= \{\, i : e_{1,i}\neq 0,\ e_{2,i}=0 \,\},\\
I_2 &= \{\, i : e_{1,i}=0,\ e_{2,i}\neq 0 \,\},\\
I_3 &= \{\, i : e_{1,i}\neq 0,\ e_{2,i}\neq 0 \,\}.
\end{align*}
We denote
\[
i_1=|I_1|,\qquad i_2=|I_2|,\qquad i_3=|I_3|.
\]

The classification reflects how each coordinate contributes to the sum-rank
weight:
\begin{itemize}[leftmargin=1.5em]
  \item at $i\in I_1$ only the first component is erroneous, and the associated
        $2\times 2$ matrix has rank $2$;
  \item at $i\in I_2$ only the second component is erroneous, and the associated
        matrix again has rank $2$;
  \item at $i\in I_3$ both components are nonzero, and the associated matrix
        has rank $1$.
\end{itemize}
This gives
\begin{equation}\label{eq:weights}
\wtsr(\mathbf{e})=2i_1+2i_2+i_3.
\end{equation}
On the other hand, the Hamming weights of the components are
\[
\mathrm{wt}_{H}(\mathbf{e}_1)=i_1+i_3
\qquad\text{and}\qquad
\mathrm{wt}_{H}(\mathbf{e}_2)=i_2+i_3,
\]
since $\supp(\mathbf{e}_1)=I_1\cup I_3$ and $\supp(\mathbf{e}_2)=I_2\cup I_3$.
We assume throughout this section that
\begin{equation}\label{eq:sr-radius-assumption}
\wtsr(\mathbf{e})\ \le\ \Big\lfloor \frac{d_{\mathrm{sr}}-1}{2}\Big\rfloor,
\end{equation}
where $d_{\mathrm{sr}}$ is the minimum sum-rank distance of $\SR(C_1,C_2)$.
The following observation, although simple, plays a crucial role in
developing our main results.
\begin{lemma}
\label{lem:upper_bound}
Let $C_1$ be an $[\ell, k_1, d_1]_4$ linear code and $C_2$ be an $[\ell, k_2, d_2]_4$ linear code. The minimum sum-rank distance of the associated binary code $\SR(C_1,C_2)$ satisfies
\begin{equation}\label{eq:lower-upper}
    d_{\mathrm{sr}}(\SR(C_1,C_2)) \le 2d_1.
\end{equation}
\end{lemma}
\begin{proof}
Recall that $\SR(C_1,C_2)$ is an $\mathbb{F}_2$-linear code, so its minimum sum-rank distance is equal to the minimum sum-rank weight of its nonzero codewords:
\[
    d_{\mathrm{sr}}(\SR(C_1,C_2)) = \min_{\mathbf{c} \in \SR(C_1,C_2), \mathbf{c} \ne 0} \mathrm{wt}_{\mathrm{sr}}(\mathbf{c}).
\]
Consider the specific subcode formed by taking the zero vector in the second component. Let $\mathbf{a}_2 = \mathbf{0} \in C_2$ (which exists since $C_2$ is linear) and let $\mathbf{a}_1 \in C_1$ be a codeword of minimum Hamming weight, i.e., $\textrm{wt}_H(\mathbf{a}_1) = d_1$.
Construct the codeword
\[
    \mathbf{c}^* = \mathbf{a}_1 x + \mathbf{a}_2 x^2 = \mathbf{a}_1 x \in \SR(C_1,C_2).
\]
According to the weight identity established in
\cite[Theorem 2.2]{CCQ2025} or Eq. \eqref{eq:CCQ-weight}, the sum-rank weight is given by
\[
   \mathrm{ wt}_{\mathrm{sr}}(\mathbf{a}_1 x + \mathbf{a}_2 x^2) = 2\mathrm{wt}_H(\mathbf{a}_1) + 2\mathrm{wt}_H(\mathbf{a}_2) - 3|\mathrm{supp}(\mathbf{a}_1) \cap \mathrm{supp}(\mathbf{a}_2)|.
\]
Substituting $\mathbf{a}_2 = \mathbf{0}$, we have $\mathrm{wt}_H(\mathbf{a}_2) = 0$
and $\mathrm{supp}(\mathbf{a}_1) \cap \mathrm{supp}(\mathbf{0}) = \emptyset$, so the intersection size is 0. Thus,
\[
    \mathrm{wt}_{\mathrm{sr}}(\mathbf{c}^*) = 2\mathrm{wt}_H(\mathbf{a}_1) + 0 - 0 = 2d_1.
\]
Since $d_{\mathrm{sr}}(\SR(C_1,C_2))$ is the minimum over all nonzero codewords, it must satisfy
\[
    d_{\mathrm{sr}}(\SR(C_1,C_2)) \le \mathrm{wt}_{\mathrm{sr}}(\mathbf{c}^*) = 2d_1,
\]
which completes the proof.
\end{proof}


As in Section~\ref{sec:prelim-Hamming}, for any received word for $C_1$ with
error vector $\mathbf{e}\in\F_q^\ell$ and erasure set $J\subseteq\{1,\dots,\ell\}$, we denote
by $t$ the number of erroneous positions outside $J$ and by $r$ the number of
erasures, so that the condition $2t+r<d_1$ is exactly the uniqueness condition
of Lemma~\ref{lem:erasure} for $C_1$.

We now state the main decoding theorem and prove it in several explicit steps.

\begin{theorem}\label{thm:SR-decoding}
Let $C_1=[\ell,k_1,d_1]_4$ and $C_2=[\ell,k_2,d_2]_4$ be quaternary linear
codes, and let $d_{\mathrm{sr}}$ be the minimum sum-rank distance of
$\SR(C_1,C_2)$. Assume
\[
d_2\ \ge\ d_{\mathrm{sr}}.
\]
Suppose further that:
\begin{itemize}[leftmargin=1.5em]
  \item there is a bounded-minimum-distance (BMD) decoder for $C_2$ that
        corrects up to $\big\lfloor \tfrac{d_2-1}{2}\big\rfloor$ Hamming errors;
  \item there is an error-erasure decoder for $C_1$ with the following
        guarantee: for every received word $\mathbf{z}\in\F_4^\ell$ and every
        erasure set $J\subseteq\{1,\dots,\ell\}$, if the (unknown) error vector
        has $t$ nonzero coordinates outside $J$ and $r=|J|$ erased positions
        with $2t+r<d_1$, then the decoder outputs the (necessarily unique)
        codeword of $C_1$ within Hamming distance $t$ on the nonerased positions (cf.\ Lemma~\ref{lem:erasure}).
\end{itemize}
Then for any transmitted $\mathbf{c}(x)=\mathbf{a}_1x+\mathbf{a}_2x^2\in\SR(C_1,C_2)$
and any error $\mathbf{e}(x)$ satisfying~\eqref{eq:sr-radius-assumption}, the
following two-step procedure recovers $\mathbf{c}(x)$ uniquely from
$\mathbf{y}(x)=\mathbf{c}(x)+\mathbf{e}(x)$:
\begin{enumerate}[leftmargin=1.5em]
  \item Step 1 (Decoding $C_2$). Extract the second component
        $\mathbf{y}_2=\mathbf{a}_2+\mathbf{e}_2\in\F_4^\ell$ and apply the BMD decoder for $C_2$ to obtain
        $\mathbf{a}_2$ and $\mathbf{e}_2=\mathbf{y}_2-\mathbf{a}_2$.
  \item Step 2 (Decoding $C_1$ with erasures). Form
        \[
        \mathbf{y}'(x)=\mathbf{y}(x)-\mathbf{a}_2x^2=
        (\mathbf{a}_1+\mathbf{e}_1)x+\mathbf{e}_2x^2
        \]
        and define the erasure set
        \[
        J:=\supp(\mathbf{e}_2)=\{\,i : e_{2,i}\neq 0\,\}.
        \]
        Apply the error-erasure decoder for $C_1$ to the first component of
        $\mathbf{y}'$, treating all positions in $J$ as erasures. This yields $\mathbf{a}_1$.
\end{enumerate}
In particular, $\SR(C_1,C_2)$ is uniquely decodable up to radius
$\big\lfloor \tfrac{d_{\mathrm{sr}}-1}{2}\big\rfloor$ in the sum-rank metric.
\end{theorem}



\begin{proof}
We keep the notation introduced above:
$\mathbf{y}(x)=\mathbf{c}(x)+\mathbf{e}(x)$ with
$\mathbf{c}(x)=\mathbf{a}_1x+\mathbf{a}_2x^2$, $\mathbf{e}(x)=\mathbf{e}_1x+\mathbf{e}_2x^2$, and $I_1,I_2,I_3$,
$i_1,i_2,i_3$ as defined.

Step~1: Control of the error in $C_2$.
By~\eqref{eq:weights} and the assumption~\eqref{eq:sr-radius-assumption} we have
\[
\wtsr(\mathbf{e})=2i_1+2i_2+i_3\ \le\ \Big\lfloor \frac{d_{\mathrm{sr}}-1}{2}\Big\rfloor.
\]
From this and $\mathrm{wt}_H(\mathbf{e}_2)=i_2+i_3$ we obtain the simple chain
\begin{align*}
\mathrm{wt}_H(\mathbf{e}_2)
&= i_2+i_3\\
&\le 2i_2+i_3\\
&= \wtsr(\mathbf{e})-2i_1\\
&\le \wtsr(\mathbf{e})\\
&\le \Big\lfloor \frac{d_{\mathrm{sr}}-1}{2}\Big\rfloor.
\end{align*}
Using the assumption $d_2\ge d_{\mathrm{sr}}$, we further deduce
\[
\Big\lfloor \frac{d_{\mathrm{sr}}-1}{2}\Big\rfloor
\le \Big\lfloor \frac{d_2-1}{2}\Big\rfloor,
\]
so
\begin{equation}\label{eq:e2-bound}
\mathrm{wt}_H(\mathbf{e}_2)\ \le\ \Big\lfloor \frac{d_2-1}{2}\Big\rfloor.
\end{equation}

Step~2: Unique decoding of $C_2$.
The second component of $\mathbf{y}(x)$ is $\mathbf{y}_2=\mathbf{a}_2+\mathbf{e}_2$. By~\eqref{eq:e2-bound}, the
Hamming error vector $\mathbf{e}_2$ has weight within the guaranteed decoding radius of
the BMD decoder for $C_2$. Therefore Step~1 of the algorithm successfully and
uniquely recovers the correct codeword $\mathbf{a}_2\in C_2$ and the error vector
$\mathbf{e}_2=\mathbf{y}_2-\mathbf{a}_2$.

Step~3: Erasure pattern for $C_1$.
After Step~2 we know $\mathbf{a}_2$ and $\mathbf{e}_2$. Subtracting $\mathbf{a}_2x^2$ from $\mathbf{y}(x)$ we
form
\[
\mathbf{y}'(x)=\mathbf{y}(x)-\mathbf{a}_2x^2
=(\mathbf{a}_1+\mathbf{e}_1)x+\mathbf{e}_2x^2.
\]
We now use the support of $\mathbf{e}_2$ as an erasure pattern. Define
\[
J:=\supp(\mathbf{e}_2)=\{\,i : e_{2,i}\neq 0\,\}=I_2\cup I_3,
\]
so the number of erasures is
\[
r=|J|=|I_2|+|I_3|=i_2+i_3.
\]
In positions $i\in J$ we declare the symbol of the first component as erased
(i.e., ``untrusted''), regardless of whether an error actually occurred in
$e_{1,i}$. This is conservative but allowed: an erasure only indicates that we
do not wish to use the value $y_{1,i}$ as a constraint in decoding.

In positions $i\notin J$ we have $e_{2,i}=0$, so such indices lie in $I_1$ or
outside $I_1\cup I_2\cup I_3$. Among these, the true (unknown) errors in the
first component are precisely the indices where $e_{1,i}\neq 0$, namely the set
$I_1$. Thus the Hamming error number for decoding $C_1$ is
\[
t=|\{\,i\notin J : e_{1,i}\neq 0\,\}| = |I_1| = i_1.
\]

Step~4: Bounding $2t+r$ via the sum-rank weight.
Using the expressions for $t$ and $r$ we compute
\begin{align*}
2t+r
&=2i_1+(i_2+i_3)\\
&=(2i_1+2i_2+i_3) - i_2\\
&= \wtsr(\mathbf{e})-i_2,
\end{align*}
where we used~\eqref{eq:weights}. Since $i_2\ge 0$, it follows that
\begin{equation}\label{eq:2tplusr-bound}
2t+r\ \le\ \wtsr(\mathbf{e})\ \le\ \Big\lfloor \frac{d_{\mathrm{sr}}-1}{2}\Big\rfloor.
\end{equation}

To turn this into a bound in terms of $d_1$, we recall the upper bound
$d_{\mathrm{sr}}\le 2d_1$ from~\eqref{eq:lower-upper}. This implies
\[
\frac{d_{\mathrm{sr}}-1}{2}\ \le\ d_1-\frac{1}{2}
\quad\Rightarrow\quad
\Big\lfloor \frac{d_{\mathrm{sr}}-1}{2}\Big\rfloor\ \le\ d_1-1 \ <\ d_1.
\]
Combining this with~\eqref{eq:2tplusr-bound} yields
\begin{equation}\label{eq:2tplusr<d1}
2t+r\ <\ d_1.
\end{equation}

Step~5: Unique decoding of $C_1$ with erasures.
We now apply the error-erasure decoder for $C_1$ to the first component (the
coefficients of $x$ in $\mathbf{y}'(x)$), with erasure set $J$ and unknown errors counted
by $t$. By Lemma~\ref{lem:erasure}, the condition $2t+r<d_1$ in~\eqref{eq:2tplusr<d1}
guarantees that there is a unique codeword $\mathbf{a}_1\in C_1$
whose restriction to $\overline{J}$
is within Hamming distance $t$ of
$(\mathbf{y}'(x))_1$;
the error-erasure decoder returns this unique nearest neighbor.
 Thus Step~3 of the algorithm succeeds and recovers
$\mathbf{a}_1$ uniquely.

Step~6: Uniqueness and decoding radius.
The argument above shows that whenever $\wtsr(\mathbf{e})\le \lfloor (d_{\mathrm{sr}}-1)/2\rfloor$, Steps~1 and~3 recover the original pair $(\mathbf{a}_1,\mathbf{a}_2)$ and therefore the original codeword $\mathbf{c}(x)$ uniquely. Hence the two-step algorithm is a valid unique decoder for $\SR(C_1,C_2)$ up to radius $\lfloor (d_{\mathrm{sr}}-1)/2\rfloor$ in the sum-rank metric.
\end{proof}

\begin{remark}
If $d_1\ge d_{\mathrm{sr}}$, then $\wtH(\mathbf{e}_1)\le \wtsr(\mathbf{e})\le\lfloor(d_{\mathrm{sr}}-1)/2\rfloor \le \lfloor(d_1-1)/2\rfloor$, so the roles of $C_1$ and $C_2$ in Theorem~\ref{thm:SR-decoding} can be interchanged: first uniquely decode $C_1$, then use $\supp(\mathbf{e}_1)$ as the erasure set for $C_2$. The proof is entirely analogous.
\end{remark}

\smallskip
We conclude this section with an illustrated example.

\smallskip
\begin{example}
We fix the quadratic extension $\F_4=\{0,1,\omega,\omega^2\}$ over $\F_2$ with
$\omega^2=\omega+1$ (so $\omega^3=1$). Let the block length be $\ell=4$ and let the evaluation set be
\[
T=(t_1,t_2,t_3,t_4)=(0,1,\omega,\omega^2).
\]

Let $C_1\subseteq\F_4^{4}$ be the Reed--Solomon code of dimension $k_1=2$ obtained by evaluating all polynomials $f\in\F_4[t]$ of degree $\le 1$ at $T$:
  \[
  C_1=\{(f(t_1),f(t_2),f(t_3),f(t_4)):\deg f\le 1\}.
  \]
  Since this is MDS, its minimum distance is $d_1=n-k_1+1=4-2+1=3$.

  Let $C_2\subseteq\F_4^{4}$ be the constant code (the Reed-Solomon code of dimension $k_2=1$), i.e.,
  \[
  C_2=\{(a,a,a,a):a\in\F_4\}.
  \]
Its minimum distance is $d_2=4$.

The associated sum-rank code $\mathrm{SR}(C_1,C_2)$ over $\F_2$ has block length $\ell=4$, matrix size $2\times2$, and dimension $2(k_1+k_2)=6$. From general bounds we have
\[
\min\{d_1,2d_2\}=3\ \le\ d_{\mathrm{sr}}(\mathrm{SR}(C_1,C_2))\ \le\ 2d_1=6.
\]

Choose
\[
f_1(t)=1+\omega t\quad\Rightarrow\quad
\mathbf{a}_1=(f_1(0),f_1(1),f_1(\omega),f_1(\omega^2))=(1,\,1+\omega,\,\omega,\,0),
\]
and choose $\mathbf{a}_2=(\omega,\omega,\omega,\omega)\in C_2$.
The transmitted sum-rank codeword is
\[
\mathbf{c}(x)=\mathbf{a}_1x+\mathbf{a}_2x^2.
\]

Let the error be
\[
\mathbf{e}_1=(0,0,1,0),\qquad
\mathbf{e}_2=(0,0,\omega,0),
\]
so that the received word is $\mathbf{y}(x)=\mathbf{c}(x)+\mathbf{e}(x)$ with $\mathbf{e}(x)=\mathbf{e}_1x+\mathbf{e}_2x^2$.
Coordinate-wise:
\[
\mathbf{y}_1=\mathbf{a}_1+\mathbf{e}_1=(1,\,1+\omega,\,\omega+1,\,0),\qquad
\mathbf{y}_2=\mathbf{a}_2+\mathbf{e}_2=(\omega,\,\omega,\,0,\,\omega).
\]

Here
\[
I_1=\{i:\ e_{1,i}\ne0,\ e_{2,i}=0\}=\varnothing,\quad
I_2=\{i:\ e_{1,i}=0,\ e_{2,i}\ne0\}=\varnothing,\quad
I_3=\{i:\ e_{1,i}\ne0,\ e_{2,i}\ne0\}=\{3\}.
\]
Hence $i_1=i_2=0$, $i_3=1$, and by the $2\times2$ rank table
\[
\wtsr(\mathbf{e})=2i_1+2i_2+i_3=1.
\]
Indeed, at coordinate $3$ the local linearized map is $L(x)=1\cdot x+\omega x^2$, which has the nonzero root $x=\omega^{-1}=\omega^2$, so its associated $2\times2$ binary matrix has rank~$1$.


Now we decode $\mathrm{SR}(C_1,C_2)$ by the two-step procedure as follows.

Step 1: Decode $C_2$.
We must uniquely decode $\mathbf{y}_2=(\omega,\omega,0,\omega)$ in $C_2$ up to $\lfloor(d_2-1)/2\rfloor=\lfloor 3/2\rfloor=1$ Hamming error.
The nearest constant vector is
\[
\mathbf{a}_2=(\omega,\omega,\omega,\omega),\qquad
\mathbf{e}_2=\mathbf{y}_2-\mathbf{a}_2=(0,0,\omega,0),
\]
with $\mathrm{wt}_H(\mathbf{e}_2)=1\le 1$, so the $C_2$ decoder outputs $\mathbf{a}_2$ and $\mathbf{e}_2$ correctly.

Step 2: Decode $C_1$ with erasures.
Form
\[
\mathbf{y}'(x)=\mathbf{y}(x)-\mathbf{a}_2x^2=(\mathbf{a}_1+\mathbf{e}_1)x+\mathbf{e}_2x^2.
\]
Use the erasure set $J=\supp(\mathbf{e}_2)=\{3\}$ (so $r=|J|=1$). The unknown Hamming errors outside $J$ in the first component are counted by
\[
t=\big|\{i\notin J: e_{1,i}\ne 0\}\big|=0,
\]
hence $2t+r=1<d_1=3$ and uniqueness is guaranteed.

Concretely, we must recover $\mathbf{a}_1\in C_1$ from $\mathbf{y}_1=\mathbf{a}_1+\mathbf{e}_1$ knowing that coordinate $3$ is erased and $t=0$ elsewhere. Since $C_1$ consists of evaluations of linear polynomials $f(t)=\alpha+\beta t$ at $T=(0,1,\omega,\omega^2)$, we interpolate $\alpha,\beta$ from the three nonerased positions:
\[
f(0)=\alpha=1,\qquad f(1)=\alpha+\beta=1+\omega\ \Rightarrow\ \beta=\omega.
\]
Check the third constraint, consistency on the fourth coordinate:
\[
f(\omega^2)=\alpha+\beta\omega^2=1+\omega\omega^2=1+\omega^3=1+1=0,
\]
which matches the observed $y_{1,4}=0$. Therefore the unique solution is
\[
\mathbf{a}_1=(f(0),f(1),f(\omega),f(\omega^2))=(1,\,1+\omega,\,\omega,\,0),
\]
exactly the transmitted $\mathbf{a}_1$.

Finally output the decoding result. The decoder returns
\[
\widehat{\mathbf{c}}(x)=\mathbf{a}_1x+\mathbf{a}_2x^2=(1,1+\omega,\omega,0)\,x+(\omega,\omega,\omega,\omega)\,x^2,
\]
which equals the transmitted codeword.
\end{example}

\begin{remark}
We had $\wtsr(\mathbf{e})=1$, whereas $d_{\mathrm{sr}}(\mathrm{SR}(C_1,C_2))\ge \min\{d_1,2d_2\}=3$, hence
\[
\wtsr(\mathbf{e})\le \Big\lfloor\frac{d_{\mathrm{sr}}-1}{2}\Big\rfloor.
\]
Moreover, $\mathrm{wt}_H(\mathbf{e}_2)=1\le \lfloor(d_2-1)/2\rfloor$ so Step~1 succeeds, and $2t+r=1<d_1=3$ so Step~2 succeeds. This illustrates the complete two-step decoding on a fully explicit instance. Thus this example is within the guaranteed radius.
\end{remark}

\section{Algorithm and Complexity}
As shown in Section 3, Theorem \ref{thm:SR-decoding} provides the theoretical foundation for the decoding strategy described in Algorithm~1.
Specifically, Theorem \ref{thm:SR-decoding} proves that, under the assumption $d_2 \ge d_{\mathrm{sr}}$ and for any error vector
$e(x)$ satisfying $w_{\mathrm{sr}}(e)\le \left\lfloor (d_{\mathrm{sr}}-1)/2 \right\rfloor$, the transmitted
codeword $c(x)=a_1x+a_2x^2 \in \textrm{SR}(C_1,C_2)$ can be uniquely recovered by a two-step procedure.
Algorithm~1 is precisely the constructive realization of this result. In the first step, the algorithm applies a
bounded-minimum-distance decoder to the second component $y_2=a_2+e_2$ in order to recover $a_2$ and
the error vector $e_2$. In the second step, the support of $e_2$ is used as an erasure set for the first
component, and an error-erasure decoder for $C_1$ is invoked to recover $a_1$. The correctness of this
procedure follows directly from the inequality $2t+r<d_1$, which is derived in the proof of Theorem \ref{thm:SR-decoding}
from the bound on the sum-rank weight of the error. Consequently, Theorem \ref{thm:SR-decoding} establishes the correctness
and decoding radius of the algorithm, while Algorithm~1 provides an explicit implementation of the
two-step erasure-guided decoding method for $\textrm{SR}(C_1,C_2)$. See the statement and proof of
Theorem \ref{thm:SR-decoding} and the description of Algorithm~1.


\begin{algorithm}[t]
\caption{Two-Step Erasure-Guided Decoding for $\SR(C_1,C_2)$}
\label{alg:twostep}
\begin{algorithmic}[1]
\State \textbf{Input:} $\mathbf{y}=(\mathbf{y}_1,\mathbf{y}_2)\in \F_4^\ell\times \F_4^\ell$; decoders for $C_1$ and $C_2$.
\State \textbf{Assume:} $d_2\ge d_{\mathrm{sr}}$ for the target code.
\State \textbf{Step 1 (Decode $C_2$):} Run a BMD decoder on $\mathbf{y}_2$ to get $\mathbf{c}_2$ and $\mathbf{e}_2=\mathbf{y}_2-\mathbf{c}_2$.
\State \textbf{Step 2 (Erasures):} Set $J\gets \supp(\mathbf{e}_2)$.
\State \textbf{Step 3 (Decode $C_1$):} Decode $\mathbf{y}_1$ in $C_1$ using $J$ as erasures to obtain $\mathbf{c}_1$.
\State \textbf{Output:} $\hat{\mathbf{c}}(x)=\mathbf{c}_1x+\mathbf{c}_2x^2$.
\end{algorithmic}
\end{algorithm}

Let $T_2$ be the complexity of the $C_2$ BMD decoder and $T_1$ the complexity of the $C_1$ error--erasure decoder. The only asymptotically nontrivial work is one call to each decoder; all other operations (forming $\mathbf{y}'$, computing $\supp(\mathbf{e}_2)$, simple bookkeeping) are linear in $\ell$ and negligible. Hence
\[
T_{\mathrm{sum\text{-}rank}}=T_2+T_1+O(\ell)\approx T_2+T_1.
\]
For BCH or Goppa codes over $\F_4$ of length $\ell$, both decoders can be implemented in $O(\ell^2)$ operations (e.g.,~\cite{HuffmanPless,Sugiyama}), yielding
\[
T_{\mathrm{sum\text{-}rank}}(\ell)=O(\ell^2).
\]

We next compare our decoder with that of Chen-Cheng-Qi
Recall that the fast decoder of Chen-Cheng-Qi~\cite{CCQ2025} first decodes $C_2$ and then runs three bounded-minimum-distance decoders for $C_1$ or its scalar multiples $u\cdot C_1$ and $u^2\cdot C_1$. If $T_2(\ell)$ and $T_1(\ell)$ denote the complexities of the decoders for $C_2$ and $C_1$, respectively, their algorithm has overall cost
\[
 T_{\mathrm{CCQ}}(\ell) = T_2(\ell) + 3 T_1(\ell).
\]
By contrast, Algorithm~1 invokes the $C_2$ decoder once and an error/erasure decoder for $C_1$ once, so that
\[
 T_{\mathrm{SR}}(\ell) = T_2(\ell) + T_1(\ell) + O(\ell).
\]
In particular,
\[
 T_{\mathrm{SR}}(\ell) = T_{\mathrm{CCQ}}(\ell) - 2 T_1(\ell) + O(\ell),
\]
and for families where $T_1(\ell)$ and $T_2(\ell)$ grow quadratically in~$\ell$ (for example, for BCH or Goppa instantiations over $\mathbb{F}_4$) the difference $2T_1(\ell)$ is already of order~$\ell^2$. Writing
\[
 T_j(\ell) = \kappa_j \ell^2 + O(\ell), \qquad j\in\{1,2\},
\]
we obtain
\[
 T_{\mathrm{CCQ}}(\ell) = (3\kappa_1 + \kappa_2)\,\ell^2 + O(\ell),\qquad
 T_{\mathrm{SR}}(\ell) = (\kappa_1 + \kappa_2)\,\ell^2 + O(\ell).
\]
Thus the leading constant in front of $\ell^2$ is reduced from $3\kappa_1+\kappa_2$ to $\kappa_1+\kappa_2$. In the common situation where the decoders for $C_1$ and $C_2$ have comparable cost (so that $\kappa_1\approx \kappa_2$), the contribution of the $C_1$-decoder to the overall complexity is essentially reduced by a factor of~three, and the total quadratic constant is roughly halved.

Additionally, consider the design region in the $(d_1,d_2)$-plane.
Let $d_{\mathrm{sr}}$ denote the minimum sum-rank distance of $\textrm{SR}(C_1,C_2)$. The decoder of~\cite{CCQ2025} guarantees unique decoding up to $\lfloor(d_{\mathrm{sr}}-1)/2\rfloor$ whenever
\[
 d_2 \ge d_{\mathrm{sr}}
 \qquad\text{and}\qquad
 d_1 \ge \frac{2}{3}d_{\mathrm{sr}}.
\]
Our Theorem~\ref{thm:SR-decoding} removes the second inequality: under the single condition $d_2 \ge d_{\mathrm{sr}}$ the two-step erasure-guided procedure still decodes up to half the sum--rank distance. In other words, for fixed $d_{\mathrm{sr}}$ the admissible region for the pair $(d_1,d_2)$ is enlarged from
\[
 \{(d_1,d_2): d_2 \ge d_{\mathrm{sr}},\ d_1 \ge \tfrac{2}{3}d_{\mathrm{sr}}\}
\]
to
\[
 \{(d_1,d_2): d_2 \ge d_{\mathrm{sr}}\}.
\]

To make the enlargement more transparent, it is convenient to normalize by $d_{\mathrm{sr}}$ and write
\[
 \delta_1 = \frac{d_1}{d_{\mathrm{sr}}}, \qquad
 \delta_2 = \frac{d_2}{d_{\mathrm{sr}}}.
\]
From the general upper bound $d_{\mathrm{sr}}(\textrm{SR}(C_1,C_2)) \le 2d_1$ we always have $\delta_1\ge \tfrac12$. In these normalized coordinates the condition of~\cite{CCQ2025} reads
\[
 \delta_2 \ge 1,\qquad \delta_1 \ge \frac{2}{3},
\]
whereas Theorem~\ref{thm:SR-decoding} only requires
\[
 \delta_2 \ge 1,\qquad \delta_1 \ge \frac{1}{2}.
\]
Thus there is an entire strip $1/2 \le \delta_1 < 2/3$ for which $\textrm{SR}(C_1,C_2)$ falls within the guarantee of our decoder (provided $d_2\ge d_{\mathrm{sr}}$),
but not within that of the decoder in ~\cite{CCQ2025}.
Equivalently, for a given target sum--rank distance $d_{\mathrm{sr}}$ one may take
\[
 \frac{1}{2}d_{\mathrm{sr}} \le d_1 < \frac{2}{3}d_{\mathrm{sr}}
\]
without sacrificing unique decodability of $\textrm{SR}(C_1,C_2)$ up to half the minimum distance. This extra freedom allows $C_1$ to be chosen with smaller Hamming distance (and hence potentially higher dimension) than was permitted in~\cite{CCQ2025}, yielding a significantly larger design space for the pair $(d_1,d_2)$.

Finally, note that in many constructions the exact value of $d_{\mathrm{sr}}$ is not known a priori. In this situation the inequality $d_2 \ge d_{\mathrm{sr}}$ can be enforced by the simple sufficient condition
\[
 d_2 \ge 2d_1,
\]
since the bound $d_{\mathrm{sr}}(\textrm{SR}(C_1,C_2))\le 2d_1$ always holds. Hence choosing $C_1$ and $C_2$ so that $d_2\ge 2d_1$ automatically places $(d_1,d_2)$ inside the enlarged design region where our two-step decoder is guaranteed to succeed up to its full unique decoding radius, without the need to compute $d_{\mathrm{sr}}$ explicitly.


\section{On the Optimality of the Two-Step Reduction}

We show that the two-step erasure-guided reduction of Theorem~\ref{thm:SR-decoding} is essentially optimal in a natural model where the decoding cost for $\SR(C_1,C_2)$ is measured solely via access to Hamming decoders for $C_1$ and $C_2$.

\subsection{From sum-rank decoding to Hamming decoding}
Consider the subcode $S_1:=\{\mathbf{a}_1x: \mathbf{a}_1\in C_1\}\subseteq \SR(C_1,C_2)$.
Since $\wtsr(\mathbf{a}_1x)=2\, \mathrm{wt}_H(\mathbf{a}_1)$ for every $\mathbf{a}_1\neq 0$, we have
$d_{\mathrm{sr}}(S_1)=2d_1$.
Now fix an SR decoder for $\SR(C_1,C_2)$ that guarantees unique decoding up to
$\big\lfloor\frac{d_{\mathrm{sr}}(\SR(C_1,C_2))-1}{2}\big\rfloor$.
For an error of the form $\mathbf{e}(x)=\mathbf{e}_1x$ we get
$\wtsr(\mathbf{e})=2\mathrm{wt}_H(\mathbf{e}_1)$; hence this SR guarantee implies
\[
\mathrm{wt}_H(\mathbf{e}_1)\le \Big\lfloor\frac{d_{\mathrm{sr}}(\SR(C_1,C_2))-1}{4}\Big\rfloor.
\]
In particular, if we restrict to families where $d_{\mathrm{sr}}(\SR(C_1,C_2))=2d_1$
(respectively $=2d_2$), then any such SR decoder necessarily solves unique
Hamming decoding for $C_1$ (respectively $C_2$) up to
$\big\lfloor\frac{d_1-1}{2}\big\rfloor$.

\subsection{Optimality under black-box Hamming decoders}
In Theorem~\ref{thm:SR-decoding}, our algorithm calls the $C_2$ BMD decoder once and the $C_1$ error-erasure decoder once. Hence, up to lower-order terms,
\[
T_{\mathrm{sum\text{-}rank}}(\ell)=T_2(\ell)+T_1(\ell)+O(\ell)\approx T_2(\ell)+T_1(\ell).
\]
Since any method that uniquely decodes $\SR(C_1,C_2)$ up to half the minimum distance must also decode the subcodes $S_1$ and $\{\mathbf{a}_2x^2\}$ to the corresponding radii, one cannot hope to reduce the worst-case complexity asymptotically below the cost of decoding $C_1$ or $C_2$ themselves without improving the underlying Hamming decoders. In this black-box model, our two-step reduction is therefore optimal up to constant factors.

\begin{remark}
Our optimality statement concerns:
\begin{itemize}[leftmargin=1.5em]
  \item unique decoding up to half the minimum sum-rank distance;
  \item worst-case (rather than average-case) complexity;
  \item a model where access to $C_1$ and $C_2$ is via black-box Hamming decoders with costs $T_1(\ell)$ and $T_2(\ell)$.
\end{itemize}
Relaxing any of these requirements may permit alternative approaches that trade decoding radius, average complexity, structure, or implementability in different ways; see, e.g.,~\cite{MP2021} for related developments in the sum-rank setting.
\end{remark}

\section{Conclusion}
We presented a concise two-step reduction that uniquely decodes $\SR(C_1,C_2)$ up to half its minimum sum-rank distance by combining one BMD decoding of $C_2$ with one error/erasure decoding of $C_1$. The method removes the $d_1\ge \tfrac{2}{3}d_{\mathrm{sr}}$ requirement, retains $O(\ell^2)$ complexity for BCH/Goppa instantiations, and enlarges the design space for binary $2\times 2$ sum-rank-metric codes.



\end{document}